\documentclass[journal,twoside]{IEEEtran}
\usepackage{amsmath,amssymb,graphicx,booktabs,array,float,url,xcolor,placeins}
\usepackage{algorithm}
\usepackage{algpseudocode}
\usepackage{tikz}
\usetikzlibrary{arrows.meta,positioning,fit,calc}
\usepackage{cite}
\usepackage[hidelinks]{hyperref}
\newtheorem{proposition}{Proposition}
\newtheorem{lemma}{Lemma}
\newtheorem{theorem}{Theorem}

\newcommand{\TableBodyFont}{\normalsize}
\begin{document}

\title{Probe-Conditioned Memory for Actuator-Deadband-Aware Koopman MPC in Industrial Sealing}

\author{Yue~Wu%
\thanks{Yue Wu is with the School of Automation, Xi'an Jiaotong University, Xi'an 710049, China, and also with Xinjiang Cigarette Factory, Hongyun Honghe Tobacco (Group) Co., Ltd., Urumqi 830000, China. Corresponding author: Yue Wu (e-mail: wuyue0619@stu.xjtu.edu.cn).}}

\markboth{IEEE Transactions on Industrial Informatics,~Vol.~XX, No.~XX, 2026}
{Probe-Conditioned Memory for Actuator-Deadband-Aware Koopman MPC}

\maketitle

\begin{abstract}
Industrial sealing and dispensing cells often reuse a pressure chain, nozzle, substrate path, and vision interface across product recipes. For a narrow bead recipe, however, a calibrated static pressure can remain correct while small corrective moves are absorbed by actuator deadband; delivered pressure changes only after a direction- and history-dependent threshold is crossed. Commissioning is defined here as the target setup and retuning interval after such a recipe change. A physical gluing and dispensing cell provides pressure-to-width calibration, a fixed probing sequence, signal-interface limits, residual scales, and actuator bounds. The controller comparison is then run on an anonymized digital twin calibrated from those measurements. The actuator-deadband-aware Koopman model predictive controller (AK-MPC) initializes from probe-conditioned memory (PCM) that links the pressure setpoint to probe-inferred actuator behavior, a predictor, a controller prior, and a fallback filter. During commissioning, a sixteen-move probe selects a nearby historical case, fits the current pressure-width relation, updates a small local dynamic correction, and supplies a feasible receding-horizon pressure policy. In the main \(1.00\) mm benchmark, where delivered-pressure loss is visible in the probe, AK-MPC reaches 0.0487 mm tracking mean absolute error (MAE) over 60 paired cases; the calibration-only inverse, adaptive proportional-integral, online recursive-least-squares ARX, and probe-fitted ARX controllers range from 0.2492 to 0.3956 mm. This large gap reflects the full constrained Koopman-MPC and online-correction workflow. The isolated PCM contribution is measured by ablation: removing PCM raises the error to 0.0655 mm. In this regime, a short actuator characterization makes historical runs useful before much target data are available.
\end{abstract}


\section{Introduction}

Industrial sealing, gluing, and dispensing lines convert a pressure or extrusion command into a geometric quality variable such as seal width or glue-line width. A pressure actuator imposes material flow, substrate and tooling conditions shape the deposited bead, an inspection device measures width, and the controller updates the next pressure move. Vision-guided sealant dispensing and force/vision monitoring have made the sealing signal path increasingly measurable \cite{hanh2022sealantvisualguidance,pereira2023sealingmonitoring}. Closed-loop gluing and robotic gluing digital twins further show how measurement and control can be integrated around the deposited bead \cite{prezas2024gluingclosedloop,kang2025roboticgluingdigitaltwin}. Human-in-the-loop commissioning studies and recent sealing data reviews indicate that deployment still depends on recipe-specific setup and production records \cite{taets2025gluedispensingcommissioning,sarti2026roboticsealingreview}. Here, commissioning denotes the recipe-specific controller setup and retuning interval after a recipe change. The harder deployment problem is controller initialization during that interval, when the static bead map is only one part of the new operating regime and target data remain scarce.

Low-width recipes expose a practical actuator nonlinearity. In proportional pneumatic valves, spool overlap can create a dead zone in which a range of command positions gives little or no flow \cite{valdiero2008deadzone}. Compact pneumatic platforms and bipolar-pressure systems further show that inflation/deflation, pressure polarity, and valve switching are part of practical pressure-actuation dynamics \cite{tian2023openpneu,mei2026bipneu}. Hysteresis compensation studies show that direction reversals and previous commands must be modeled when the actuator has memory \cite{park2024tcnhysteresis,shen2025fprc}.

The practical difficulty is not only model accuracy, but the timing of adaptation. Every product recipe specifies a target bead width, and calibration gives the pressure expected to deliver it. A move from a nominal wider recipe to a narrow one can also alter delivered-pressure loss, sensor bias, lag, hysteresis, and drift. Static inverse control keeps the calibrated setpoint but ignores the pressure that fails to reach the bead. Autoregressive or recursive-least-squares identification can learn the new dynamics, although it spends scarce target samples rediscovering actuator behavior already present in similar recipes. Black-box sequence models fit historical trajectories, yet the pressure bounds, setpoints, and run records that a controller needs at deployment are usually not preserved in a usable form. The missing object is probe-conditioned memory (PCM): a record that combines static calibration with a short actuator probe.

Model predictive control (MPC) is a natural execution layer when pressure, rate, and quality constraints matter \cite{rawlings2017,qin2003}. Online learning and scheduled MPC reduce mismatch during operation, but their initialization still depends on data gathered after the recipe changes \cite{zheng2022onlinelearningdriftmpc,hu2022scheduledmodempc}. Sparse adaptive modeling and online linearization address changing nonlinear dynamics once observations are available \cite{abdullah2023adaptivesparsempc,luo2023onlinelinearizationmpc}. Switched nonlinear predictive control and recurrent predictive models expand the online operating range \cite{hu2024switchednonlinearmpc,alhajeri2022structuredrnnmpc}. Koopman lifting provides compact nonlinear predictors for constrained control \cite{williams2015,korda2018}, and hybrid or multi-model variants improve representation across regimes \cite{son2022hybridkoopmanmpc,lawrynczuk2024koopmanmultimodel}. Experimental Koopman predictive control has also been demonstrated in process-control loops \cite{valabek2026deepkoopmanpasteurization}. Case-based reasoning and transfer learning supply general reuse mechanisms across tasks \cite{aamodt1994,pan2010}. Digital twins and dynamic transfer soft sensors carry process information across operating conditions \cite{zhu2022digitaltwinquality,zhang2023dynamictransfersoftsensor}. Transfer and federated predictive control extend this idea to nonlinear process settings \cite{xiao2023transfermpc,xu2024federatedmpc}. In the low-width sealing case, reuse must carry an actuator-delivery state into the first commissioning run, not just a predictor checkpoint.

The main comparison uses a narrow, one-millimeter target bead width because this setpoint is feasible from static calibration alone but sensitive to pressure-delivery loss. A calibration-based inverse controller, a calibration-based adaptive proportional-integral (PI) controller, a probe-identified autoregressive (ARX) dynamic controller, and an online recursive-least-squares (RLS) identification-based controller serve as calibration-aware, feedback, and dynamic-identification baselines. The probe-fitted ARX arm is included specifically as a probe-information/no-memory baseline; it uses the target probe to identify local dynamics but does not retrieve a stored PCM record.

The actuator-deadband-aware Koopman MPC starts from such a record. Probe-conditioned memory stores the calibrated pressure setpoint, material and probe response, actuator-delivery estimate, predictor state, and controller warm start as a single context; Section~III gives the tuple. The controller injects the target pressure-width fit, updates a local dynamic correction, and filters pressure moves through receding-horizon constraints. The analysis explains why a calibrated pressure root is not a sufficient state in the deadband regime and relates probe error, retrieval distance, local correction error, and finite-horizon tracking cost. The evaluation uses a calibrated one-millimeter benchmark with twelve target conditions, five random seeds, paired baselines, mechanism ablations, and safety/runtime summaries.

\section{Plant-Calibrated Sealing Digital Twin and Problem Formulation}

The physical gluing/dispensing cell in Fig.~\ref{fig:physical-unit} fixes the variables used throughout the paper. The manipulated variable is the pressure increment $u_k=\Delta p_k$, implemented through a pneumatic pressure/valve actuator. The process pressure $p_k$ drives the glue applicator nozzle, and the quality variable is the coating width $w_k$ measured by a top-view vision camera. The same signal path supplies calibration data, probing responses, actuator limits, and residual statistics. The digital-twin state $x_k=[p_k,w_k]^\top$, the pressure-root target, and the online adaptation buffer remain tied to measured actuator and vision signals.

\begin{figure*}[!t]
\centering
\includegraphics[width=0.95\textwidth]{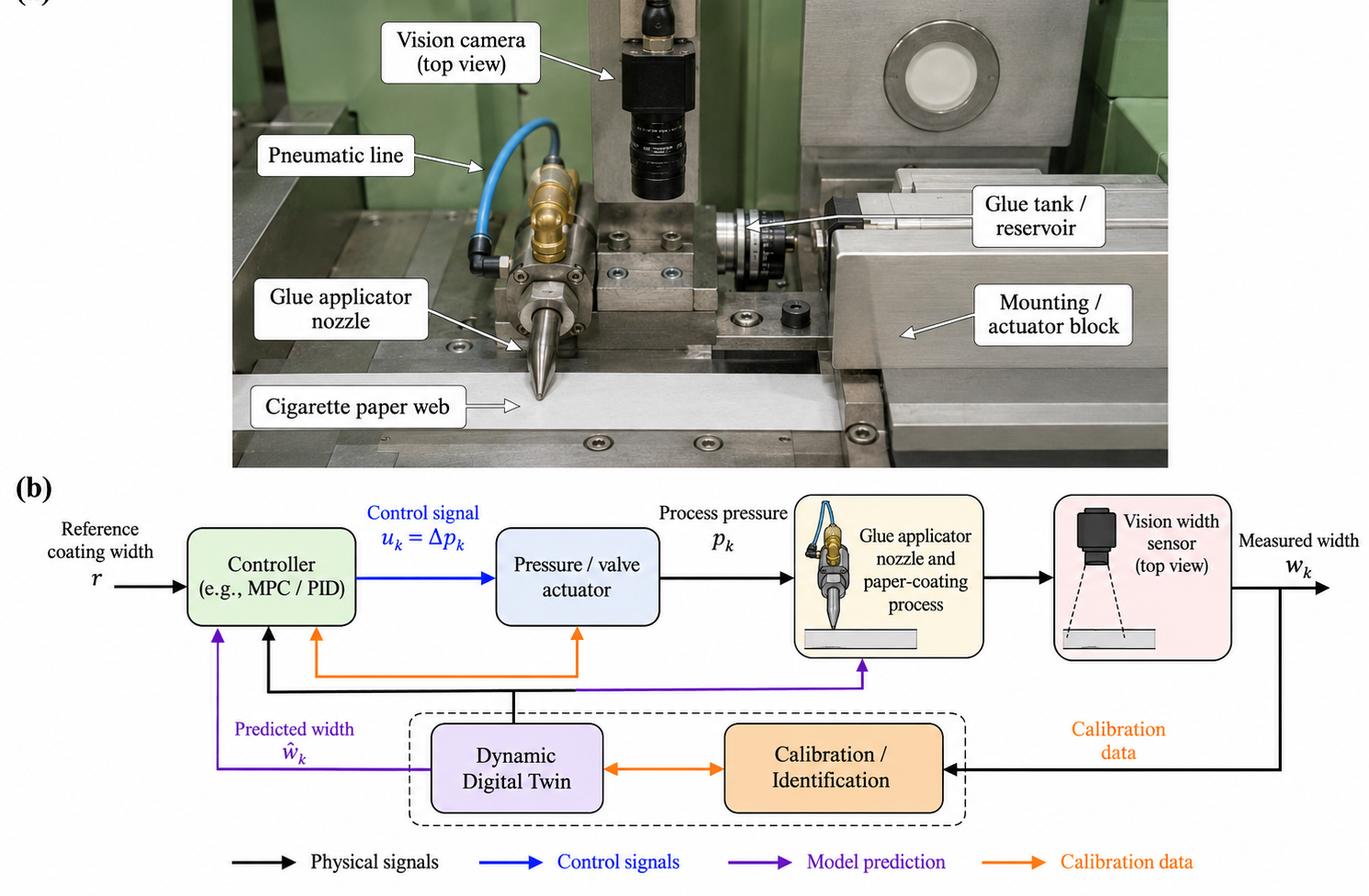}
\caption{Physical gluing/dispensing cell and surrogate-control signal flow. (a) Physical test cell with pneumatic pressure actuation, glue applicator nozzle, paper web, top-view vision camera, glue reservoir, and mounting/actuator block. (b) Signal flow used in the plant-calibrated control study. The physical cell supplies calibration, probe, actuator-limit, and measurement-residual data; the anonymized dynamic twin then executes the paired controller comparison.}
\label{fig:physical-unit}
\end{figure*}

Three physical measurements set the scale of the study. A pressure-width calibration experiment fits the low-width recipe map and pressure root. Signal-interface tests record pressure/rate limits and measurement residuals. A fixed 16-action probing experiment gives the short transient response used to estimate deadband, delivered-pressure gain, lag, and hysteresis. These measurements define the calibration envelope for an anonymized dynamic twin; the paired controller comparison is a calibrated digital-twin study.

The pressure-width calibration used for the displayed stress recipe is monotone and saturating over the tested pressure range (Fig.~\ref{fig:physical-calibration}). For each recipe \(b\), the fitted calibration supplies a recipe pressure root \(p_b^r\) from \(\hat f_b(p_b^r)=r_b\). The displayed low-width stress recipe uses \(r_b=1.00\) mm and gives \(p_b^r=0.467\) bar for \(w(p)=\alpha-\beta e^{-\gamma p}\), with \(\alpha=1.62\), \(\beta=1.22\), and \(\gamma=1.45\). Historical nominal contexts and auxiliary recipe runs may use other references, including \(r_b=1.50\) mm. The calibration residual root mean square error (RMSE) is \(0.010\) mm and enters the prediction tube and actuator-deadband benchmark as a static uncertainty scale. The fixed 16-action probe supplies a short target-side dynamic, material, and actuator descriptor.

\begin{figure}[!t]
\centering
\includegraphics[width=\linewidth]{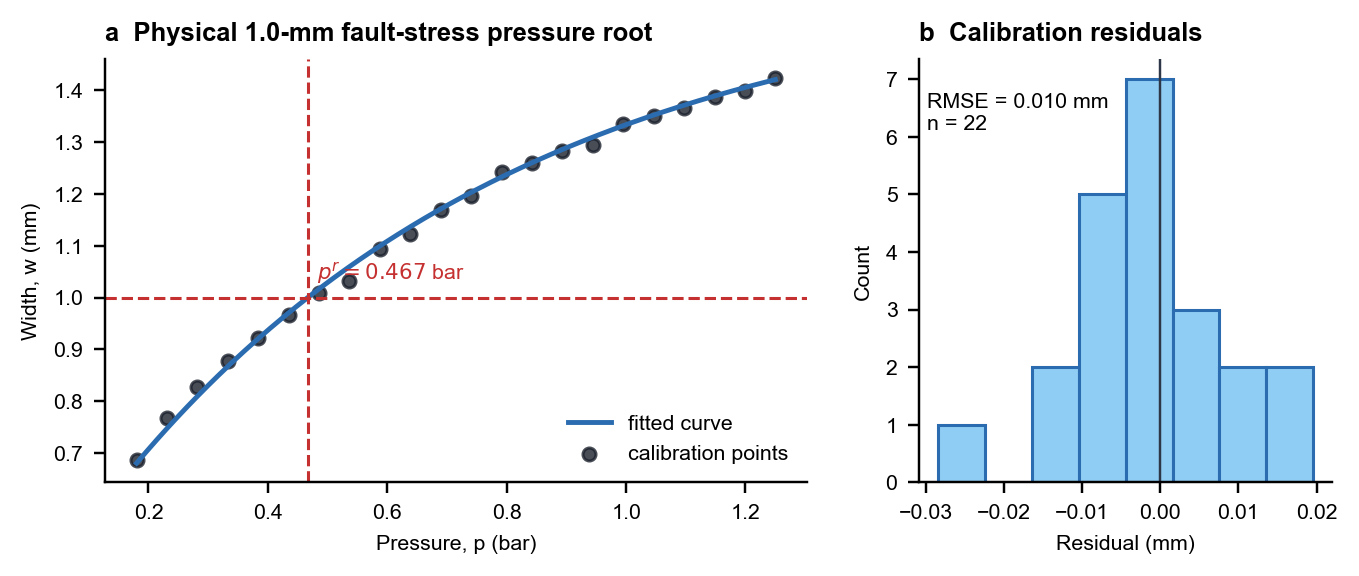}
\caption{Physical pressure-width calibration for the displayed low-width recipe. The fitted curve supplies the \(r_b=1.00\) mm pressure root used by static inverse MPC, PI, ARX-MPC, and AK-MPC. Its residual scale enters the calibration-error term in the benchmark model.}
\label{fig:physical-calibration}
\end{figure}

\begin{table*}[!t]
\centering
\caption{Physical calibration and benchmark interface used in the actuator-deadband stress study.}
\label{tab:physical-diagnostics}
\TableBodyFont
\begin{tabular}{>{\raggedright\arraybackslash}p{0.23\linewidth}>{\raggedright\arraybackslash}p{0.20\linewidth}>{\raggedright\arraybackslash}p{0.47\linewidth}}
\toprule
Quantity & Value & Use in the study \\
\midrule
Stress-recipe width & 1.00 mm & low-width reference \(r_b\) used for the main benchmark \\
Pressure root & 0.467 bar & \(p_b^r\) computed from \(w(p)=\alpha-\beta e^{-\gamma p}=r_b\) \\
Calibration model & $\alpha=1.62$, $\beta=1.22$, $\gamma=1.45$ & saturating pressure-width fit for the physical unit \\
Physical calibration residual RMSE & 0.010 mm & pressure-width fit residual; kept separate from the target-fit residual used in the surrogate cases \\
Calibration samples & 22 & samples used for the displayed residual histogram \\
Signal-interface and probe tests & private traces, public summaries & pressure actuation, width sensing, actuator limits, delivered-pressure loss, and short-probe execution \\
Production and recipe records & confidential raw records & anonymized material, drift, and actuator-response envelopes for the released benchmark \\
Recipe reference in main run & \(r_b=1.00\) mm & setpoint supplied to all controller arms \\
Benchmark pressure domain & $[0.18,1.25]$ bar & actuator pressure envelope used by all controller arms \\
Reporting band & $\tau_w=0.035$ mm & width-error band used for out-of-band reporting \\
\bottomrule
\end{tabular}
\end{table*}

Table~\ref{tab:physical-diagnostics} fixes the physical scale used by the surrogate comparison: the pressure root, calibration residual, admissible pressure range, and reporting band. Raw production records remain internal to the plant; the released benchmark uses anonymized material, drift, and actuator-response envelopes derived from that operating range.

\subsection{Physical-Cell Calibration and Probe Experiments}

Physical-cell measurements provide pressure actuation, top-view width measurement, recipe-specific calibration, actuator/rate limits, measurement residuals, and short-probe responses. The anonymized surrogate then runs the multi-controller stress tests under fixed seeds, target definitions, and shared constraints. Raw industrial records are kept confidential because they contain product recipes, material specifications, production settings, maintenance logs, and operation traces. The public benchmark exposes calibrated maps, residual scales, actuator limits, lag, hysteresis, drift, and actuator-response ranges without releasing the original plant logs.

The benchmark couples a pressure-actuated surrogate process with a probe-conditioned memory layer (Fig.~\ref{fig:system}). In the process path, actuator limits and rate limits shape the pressure increment, the calibrated plant map converts pressure history to width, and an inspection model returns measured width. In the memory path, a short target probe forms a static-dynamic query, the closest stored context initializes prediction and MPC priors, and target observations update the local dynamic correction and run record.

\begin{figure}[!t]
\centering
\includegraphics[width=0.9\linewidth]{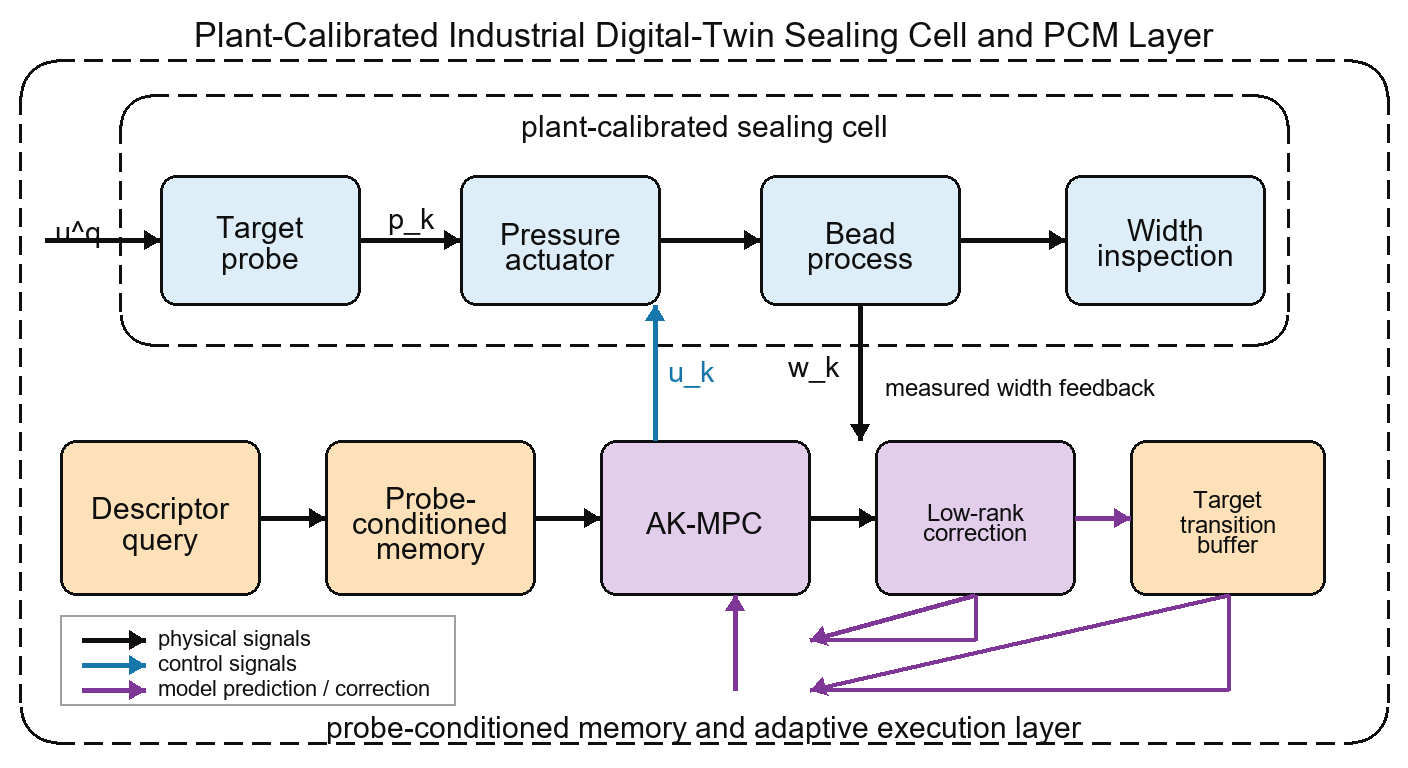}
\caption{Plant-calibrated sealing surrogate and probe-conditioned memory layer used for context retrieval, adaptive prediction, and feasible MPC execution.}
\label{fig:system}
\end{figure}

The controller uses the variables in Table~\ref{tab:variables} throughout the calibration, probing, and actuator-deadband stress runs.

\begin{table}[!t]
\centering
\caption{Industrial sealing variables and their control roles.}
\label{tab:variables}
\TableBodyFont
\begin{tabular}{@{}>{\raggedright\arraybackslash}p{0.18\linewidth}>{\raggedright\arraybackslash}p{0.24\linewidth}>{\raggedright\arraybackslash}p{0.46\linewidth}@{}}
\toprule
Variable & Role & Digital-twin implementation \\
\midrule
Recipe \(b\) & Product or brand condition & selects the target reference, material response, and admissible commissioning record \\
Reference \(r_b\) & Recipe quality target & \(1.00\) mm for the main low-width stress recipe; other recipe references are auxiliary support runs \\
Pressure root \(p_b^r\) & Static initialization & root of \(\hat f_b(p_b^r)=r_b\) used by static inverse, MPC terminal pressure, and fallback logic \\
Material descriptor \(d_b^m\) & Material state & measurable vector from calibration and probe response, kept separate from \(r_b\) and \(p_b^r\) \\
Pressure \(p_k\) & Actuator or command state & benchmark bounded to \( [0.18,1.25] \) bar in the main stress recipe \\
Action \(\Delta p_k\) & Controller move & rate-limited to \( [-0.055,0.055] \) bar per step in the benchmark \\
Width \(w_k\) & Quality output & physical coating width or scaled benchmark seal width, both measured with noise \\
Probe sequence & Target characterization & fixed 16 pressure moves used to form dynamic and actuator-response descriptors \\
Context record & Stored operating case & pressure root, material/probe response, actuator descriptor, predictor, scaler, MPC prior, fallback state \\
\bottomrule
\end{tabular}
\end{table}

The material descriptor \(d_b^m\) has seven entries derived from calibration and probe data. Calibration contributes the slope at \(p_b^r\), curvature near the pressure root, saturation margin, and physical calibration RMSE from Fig.~\ref{fig:physical-calibration} and Table~\ref{tab:physical-diagnostics}. The probe contributes transient gain, lag index, and hysteresis index. The static descriptor \(d_b^s\) stores recipe reference, pressure root, and coarse static-map position; \(d_b^m\) stores derivative, curvature, saturation, uncertainty, and transient material-response quantities. Keeping these roles separate prevents the retrieval metric from counting the same pressure-root information twice.

The physical calibration in Fig.~\ref{fig:physical-calibration} is displayed with a saturating exponential fit because that form is interpretable for the measured rig over the tested pressure range. The controller and the finite Koopman closure below use a local polynomial surrogate on the compact operating interval; the polynomial approximation residual is carried as the calibration term $\epsilon_f$ in the prediction tube. The measured pressure root and residual scale enter the controller separately from the algebraic dictionary used for tube filtering.

The dynamic surrogate used for the paired controller benchmark is
\begin{equation}
w_{k+1}=(1-\ell_\theta) f_{b,\theta,g}(p_{k+1}^{eff})+\ell_\theta w_k+h_\theta \Delta p_k+b_{\theta,e}+b_\theta^{sens}+\varepsilon_k,
\end{equation}
where \(f_{b,\theta,g}\) is the recipe-specific pressure-width map after curve-warp or local restriction effects, \(p_{k+1}^{eff}\) is the effective pressure after actuator loss or deadband, \(\ell_\theta\) is first-order width memory, \(h_\theta\) is action-dependent hysteresis, \(b_{\theta,e}\) is an episode-level drift term, \(b_\theta^{sens}\) is a sensor-bias term, and \(\varepsilon_k\) is inspection noise. The model keeps the transient terms explicit: two recipes can share similar pressure-width curves while exhibiting different delivered-pressure responses, and two targets within a recipe can share \(r_b\) while differing in material and actuator behavior.

The implementation represents delivered pressure through the deadband-gain map
\begin{align}
\mathcal D_\theta(u_k)&=\operatorname{sgn}(u_k)g_\theta[|u_k|-d_\theta]_+,\\
p_{k+1}^{eff}&=\Pi_{[p_{\min},p_{\max}]}\{p_k+\mathcal D_\theta(u_k)\},
\end{align}
with \(d_\theta\ge 0\), \(g_\theta\in(0,1]\), and \([s]_+=\max\{s,0\}\). The nonlinear runner uses the same structure with a saturating transient \(h_\theta\tanh(u_k/6)\); the term \(h_\theta\Delta p_k\) in the plant equation above is its local surrogate in the lifted analysis. Asymmetric positive/negative deadbands can be handled by replacing \(d_\theta,g_\theta\) with \(d_\theta^\pm,g_\theta^\pm\). Probe-conditioned memory lookup uses the actuator descriptor
\begin{equation}
z_\theta^a=[d_\theta,\; g_\theta,\; \ell_\theta,\; h_\theta]^\top ,
\end{equation}
with calibration residual, measurement noise, and slow drift carried in \(\rho_i\) and in the tube terms below.

These terms are tied to measured or inferred quantities. The pressure-width term \(f_{b,\theta,g}\) is anchored by noisy calibration samples, material descriptors, and the target response mode. The transient coefficients \(\ell_\theta\) and \(h_\theta\) come from probe response and pressure-reversal behavior. The variables \(b_{\theta,e}\), \(b_\theta^{sens}\), and \(p^{eff}\) describe drift, sensor bias, and delivered pressure after loss or deadband; \(\varepsilon_k\) represents inspection noise; and \(\rho_i\) records fit residuals, retrieval distance, actuator-mode labels, and run history. Here \(\theta\) indexes a target condition, while the scalar deadband \(d_\theta\) is distinct from the descriptor symbols \(d_i^s\), \(d_i^m\), and \(d_i^q\).

\begin{proposition}[Root-only indistinguishability under actuator deadband]
Consider two targets \(\theta_1,\theta_2\) with the same calibrated pressure-width map \(f_b\) and pressure root \(f_b(p_b^r)=r_b\), but with delivered-pressure maps \(g_{a,\theta_j}(p,u)\), \(j=1,2\). Let a controller choose its first pressure move using only \((\hat f_b,p_b^r)\) and the same width history, before any actuator-specific target response to that move has been observed. If \(g_{a,\theta_1}(p_k,u_k)\ne g_{a,\theta_2}(p_k,u_k)\), \(|f'_b(p)|\ge m_f>0\) on the compact operating interval, \(\ell_\theta\le\bar\ell<1\), and the combined residual difference is bounded by \(2\bar\varepsilon\), then the next-width separation satisfies
\begin{equation}
|w_{k+1}^{(1)}-w_{k+1}^{(2)}|
\ge (1-\bar\ell)m_f |g_{a,\theta_1}(p_k,u_k)-g_{a,\theta_2}(p_k,u_k)|-2\bar\varepsilon .
\end{equation}
\emph{Proof.} The information available to a root-only initializer is identical for the two targets, so it applies the same \(u_k\). Subtracting the two width updates leaves the difference between \(f_b(g_{a,\theta_1}(p_k,u_k))\) and \(f_b(g_{a,\theta_2}(p_k,u_k))\), plus bounded residuals. The mean-value theorem and the lower slope bound give the stated inequality.
\end{proposition}

\begin{proposition}[Deadband lower bound for small pressure moves]
For the delivered-pressure map \(\mathcal D_{\theta_j}(u)=g_j^+[u-d_j^+]_+\) in the positive direction, if \(d_1^+<u<d_2^+\), then target \(\theta_1\) receives \(g_1^+(u-d_1^+)\) effective pressure while target \(\theta_2\) receives none. Hence
\begin{equation}
|w_{k+1}^{(1)}-w_{k+1}^{(2)}|
\ge (1-\bar\ell)m_f g_1^+(u-d_1^+)-2\bar\varepsilon .
\end{equation}
The static pressure root is necessary but not sufficient in the low-width deadband regime: the missing state is the actuator delivery map, which the fixed probe estimates before feedback commissioning begins.
\end{proposition}

\section{Probe-Conditioned Memory}

Each stored operating case extends a model checkpoint with the pressure root, material response, short-probe response, actuator descriptor, predictor state, normalization, MPC prior, and fallback filter needed to initialize actuator-deadband-aware Koopman MPC before extensive target data are available.

\subsection{Context Record and Feasible Pressure Moves}

The closed-loop state is $x_k=[p_k,w_k]^\top$, where $p_k$ is pressure and $w_k$ is width. The manipulated input is the pressure increment $u_k=\Delta p_k$. The tube-filtered controller applies the state-dependent admissible input set
\begin{equation}
\mathcal U_p(p_k):=[u_{\min},u_{\max}]\cap[p_{\min}-p_k,p_{\max}-p_k].
\end{equation}
If $u_k\in\mathcal U_p(p_k)$ and $p_k\in[p_{\min},p_{\max}]$, then $p_{k+1}=p_k+u_k\in[p_{\min},p_{\max}]$ without activating a clip.

The fixed probe sequence \(\mathcal U^q=\{u_0^q,\ldots,u_{N_q-1}^q\}\) is used to estimate the actuator part of the target state before the first closed-loop episode. From the probe log, the implementation observes the effective pressure increment \(\Delta p_k^{eff}=p_{k+1}-p_k\) and fits
\begin{equation}
(\hat d,\hat g)=\arg\min_{d,g}\sum_{k\in\mathcal I_q}
\left(\Delta p_k^{eff}-\operatorname{sgn}(u_k^q)g[|u_k^q|-d]_+\right)^2 ,
\end{equation}
where \(d\) is searched on a short grid and \(g\) is obtained by least squares. With \(\hat p_{k+1}^{eff}\) fixed by the logged pressure, lag and transient response are estimated by
\begin{equation}
[\hat\ell,\hat h,\hat c]^\top=\arg\min_\beta\sum_{k\in\mathcal I_q}
\left(y_k-\varphi_k^\top\beta\right)^2 ,
\end{equation}
where \(y_k=w_{k+1}^{obs}-\hat f_b(\hat p_{k+1}^{eff})\) and
\(\varphi_k=[w_k-\hat f_b(\hat p_{k+1}^{eff}),\tanh(u_k^q/6),1]^\top\). The target actuator descriptor used below is \(z_*^a=[\hat d,\hat g,\hat\ell,\hat h]^\top\), while \(\hat c\) and the probe residuals are stored in the run record.

\begin{lemma}[Probe-estimation error under deadband excitation]
Suppose the calibration map is monotone on the operating interval with \(0<m_f\le |\hat f'_b(p)|\), the probe contains both pressure directions and at least one command beyond each relevant deadband by margin \(\Delta_u>0\), and the regression matrices in the two estimates above have smallest singular value bounded below by \(\mu_q>0\). If the pressure/width residual entering the probe fit is bounded by \(\bar e_q\), then for a constant \(C_q\) depending on \(\mu_q\), the probe sequence, and the local slope bound,
\begin{equation}
\|\hat z_*^a-z_*^a\|_2 \le C_q\left(\frac{\bar e_q}{m_f}+\epsilon_f\right).
\end{equation}
\emph{Proof sketch.} Monotonicity converts width residuals to effective-pressure residuals with gain no larger than \(1/m_f\). The deadband-exciting probe keeps the deadband/gain and lag/hysteresis regression matrices away from rank deficiency. Standard least-squares perturbation bounds then give the stated estimate.
\end{lemma}

Each historical operation is stored as
\begin{equation}
\mathcal C_i=(b_i,r_i,p_i^r,\mathcal D_i,d_i^s,d_i^m,d_i^q,z_i^a,\Theta_i^K,S_i,\pi_i,\rho_i),
\end{equation}
where \(b_i\) is the product recipe or brand condition, \(r_i\) is its quality reference, \(p_i^r\) is the stored pressure root, $\mathcal D_i$ contains pressure-width and transition data, $d_i^s$ is the static pressure-width descriptor, \(d_i^m\) is the material-response vector defined above, $d_i^q$ is the probing-response descriptor, \(z_i^a\) stores the observed or inferred actuator-response descriptor, $\Theta_i^K=(\phi_\eta,\psi_\eta,A_i,B_i,c_i)$ is the Koopman prediction state, $S_i$ is the state scaler, $\pi_i$ is the MPC prior, and $\rho_i$ stores fit residuals, retrieval distance, update history, fallback state, and lifecycle information. For a new target recipe \(b_\ast\), a noisy pressure-width calibration gives $\hat f_\ast(p)=\hat a_2p^2+\hat a_1p+\hat a_0$ and the pressure root $p_\ast^r$ solving \(\hat f_\ast(p_\ast^r)=r_\ast\) closest to the operating range center. The calibration supplies \(d_\ast^s\) and \(d_\ast^m\), the target probe supplies \(d_\ast^q\), observable residual/probe statistics supply \(z_\ast^a\), and retrieval uses
\begin{equation}
\delta_i(\ast)=\alpha\frac{\|d_\ast^s-d_i^s\|_2}{m_s+\epsilon_m}
+(1-\alpha)\frac{\|[d_\ast^m,d_\ast^q,z_\ast^a]-[d_i^m,d_i^q,z_i^a]\|_2}{m_q+\epsilon_m}.
\end{equation}
Here \(z^a\) is inferred from calibration and probe data. The descriptor combines four actuator-response quantities: transient gain, lag index, hysteresis asymmetry, and a deadband proxy. Residual statistics, including bias, pressure-efficiency loss when available, variance, and drift slope, are stored with the same target record. Hidden target parameters and true response-mode labels instantiate the surrogate plant and offline record; deployed controllers receive only the inferred descriptor. The closest feasible context supplies $(\Theta_i^K,S_i,\pi_i)$, while large distance, large calibration residual, or high actuator mismatch is recorded in $\rho_\ast$ as a fallback flag. In this form, probe-conditioned memory stores both the predictor and the information needed to execute a short commissioning run.

\begin{theorem}[Retrieved probe-conditioned memory bounds the initial predictor mismatch]
Let \(\xi_i=[d_i^s,d_i^m,d_i^q,z_i^a]^\top\), let \(\hat\xi_*=[d_*^s,d_*^m,d_*^q,\hat z_*^a]^\top\), and let the retrieved context be \(i^*=\arg\min_i\delta_i(*)\). Assume that the lifted one-step predictor is locally Lipschitz in the descriptor, namely
\[
\|F_K(z,u;\xi_1)-F_K(z,u;\xi_2)\|_2\le L_K\|\xi_1-\xi_2\|_2,
\]
and let \(\epsilon_K\) denote finite-dictionary closure and measurement residuals. Then the target predictor initialized from \(i^*\) satisfies
\begin{equation}
\|z_{k+1}-\hat z_{k+1}^{\,i^*}\|_2
\le L_K\bigl(\|\xi_*-\hat\xi_*\|_2+\|\hat\xi_*-\xi_{i^*}\|_2\bigr)+\epsilon_K .
\end{equation}
The initial mismatch has two terms: the probe-estimation error in Lemma~1 and the source-library covering radius around the estimated target descriptor.
\emph{Proof.} Insert and subtract \(F_K(z,u;\hat\xi_*)\) and \(F_K(z,u;\xi_{i^*})\), apply the Lipschitz condition, and add the closure residual. The retrieval term is bounded by the selected nearest-context distance in the retrieval rule above.
\end{theorem}

\section{Actuator-Deadband-Aware Koopman MPC With Feasibility Filtering}

The retrieved context becomes a target controller after four operations: inject the target recipe, pressure-width map, and pressure root; adapt the predictor with a low-dimensional local correction; search feasible receding-horizon pressure moves; and record when the learned candidate is accepted or replaced by the fallback filter.

\begin{algorithm}[!t]
\caption{Actuator-deadband-aware Koopman MPC initialized from probe-conditioned memory}
\label{alg:akmpc}
\footnotesize
\begin{algorithmic}[1]
\Require Probe-conditioned memory $\mathcal M=\{\mathcal C_i\}_{i=1}^{N_c}$; target recipe \(b_*\); recipe reference \(r_*\); target calibration pairs $\mathcal D_*^{\rm cal}$; material descriptor \(d_*^m\); fixed probe sequence $\mathcal U^q$; pressure and input bounds; horizon $H$; local-correction setting.
\Ensure Closed-loop pressure increments $u_k$ and updated commissioning record $\rho_*$.
\State Fit $\hat f_*(p)$ from $\mathcal D_*^{\rm cal}$, compute $\epsilon_f$, and find the feasible pressure root \(p_*^r\) satisfying \(\hat f_*(p_*^r)=r_*\).
\State Form $d_*^s$ from $\hat f_*$, $p_*^r$, and residual statistics; update \(d_*^m\) from curve-shape and material cues when it is not supplied; execute $\mathcal U^q$ and form $d_*^q$ and \(z_*^a\) from lag, transient gain, hysteresis, bias, deadband, efficiency, and residual statistics.
\State Retrieve $i^*=\arg\min_i\delta_i(*)$; initialize $(\Theta^K_{i^*},S_{i^*},\pi_{i^*},p_*^r)$, probe-conditioned prior, buffer $\mathcal B_*$, local correction, and run record $\rho_*$.
\For{closed-loop time $k=0,1,\ldots$}
    \State Measure $x_k=[p_k,w_k]^\top$ and append the latest transition to $\mathcal B_*$ when available.
    \If{online target updating is enabled and a new transition is available}
        \State Update the local correction used by the target predictor; log residuals and correction status in $\rho_*$.
    \Else
        \State Keep the previous correction and log frozen or low-information status.
    \EndIf
    \State Generate candidate sequences: shifted previous plan, zero move, constant moves, proportional width correction, pressure-root first-step/ramp moves, and sampled feasible perturbations.
    \State Roll out candidates with target-map injection, actuator-aware pressure delivery, and the corrected predictor; discard candidates whose predicted pressure leaves the admissible interval.
    \State Score finite candidates using width tracking, input effort, input smoothness, and enabled terminal penalties.
    \If{at least one feasible candidate exists}
        \State Apply the first move of the lowest-cost feasible candidate; optionally update the best-candidate record in $\pi_*$ for warm-starting.
    \Else
        \State Apply the bounded pressure-root or static inverse adaptive-PI fallback filtered through $\mathcal U_p(p_k)$.
    \EndIf
  \State Update the commissioning record $\rho_*$ with retrieval, residual, fallback-filter, and constraint information.
\EndFor
\end{algorithmic}
\end{algorithm}

\subsection{Feasibility Filtering and Mechanism Analysis}

The feasibility layer serves as a fallback filter for the learned controller. It records candidate rejection, input saturation, fallback activation, and prediction-tube width while keeping pressure and input constraints hard. Because the main evaluation is performed on a calibrated surrogate, the controller must still tolerate a residual simulation-to-reality gap. The online local correction and fallback filter are designed to absorb unmodeled physical mismatch, material drift, and calibration residuals when the controller is transferred from a stored context to a commissioning run. Before measurement noise is added, the plant-calibrated digital twin has the scalar structure
\begin{align}
p_{k+1}&=p_k+u_k,\\
w_{k+1}&=(1-\ell)f(p_k+u_k)+\ell w_k+h u_k+b_k+\varepsilon_k .
\end{align}
With a local quadratic pressure-width surrogate \(q(p)=a_2p^2+a_1p+a_0\), the dictionary
\begin{equation}
\phi(x_k)=[1,p_k,w_k,p_k^2]^\top
\end{equation}
is closed by an input-augmented finite dictionary.

\begin{proposition}[Finite lifted closure for the local pressure-width surrogate]
Consider \(p_{k+1}=p_k+u_k\) and
\[
w_{k+1}=(1-\ell)q(p_k+u_k)+\ell w_k+h u_k+\bar b+\epsilon_k,
\]
where \(q(p)=a_2p^2+a_1p+a_0\). With
\[
\psi_k=[1,p_k,w_k,p_k^2,u_k,p_ku_k,u_k^2]^\top,
\]
there exists a finite matrix \(M_\theta\), depending on \(a_0,a_1,a_2,\ell,h,\bar b\), such that \(\phi(x_{k+1})=M_\theta\psi_k+r_{k+1}\). The residual \(r_{k+1}\) collects the difference between the local polynomial and the calibrated pressure-width map, drift variation \(b_k-\bar b\), inspection noise, and unmodeled hysteresis or actuation effects.
\end{proposition}

Indeed, \(p_{k+1}\), \(p_{k+1}^2=p_k^2+2p_ku_k+u_k^2\), and the quadratic expansion of \(q(p_k+u_k)\) are all linear in \(\psi_k\); only the calibrated-map approximation and nonstationary physical terms remain outside the finite closure. The lifted predictor is built around the physical pressure-width interface instead of a generic sequence model.

For target transfer, the retrieved source operator is corrected as
\begin{equation}
z_{k+1}=(A_i+UV^\top)z_k+(B_i+\Delta B)u_k+c_i+\Delta c .
\end{equation}
With a transition buffer \(\mathcal B_k=\{(z_t,u_t,z_{t+1})\}_{t=k-M}^{k-1}\), define the retrieved-source residuals
\[
R_t=z_{t+1}-A_i z_t-B_i u_t-c_i,\qquad
\Omega=[Z_-^\top,\;U_-^\top,\;\mathbf 1]^\top .
\]
The regularized finite-rank update used by the controller can be written as
\begin{align}
\Delta\Theta_{\rm ridge}&=R\Omega^\top(\Omega\Omega^\top+\lambda_\Delta I)^{-1},\\
\widehat{\Delta\Theta}_{r_\Delta}&=\mathcal T_{r_\Delta}(\Delta\Theta_{\rm ridge}),
\end{align}
where \(\mathcal T_{r_\Delta}\) denotes rank-\(r_\Delta\) truncation and \(\widehat{\Delta\Theta}_{r_\Delta}=[\Delta A,\Delta B,\Delta c]\). The low-rank form reflects the sealing structure. The pressure update is the same kinematic row \(p_{k+1}=p_k+u_k\) across contexts, and the pressure-root injection is handled by the static calibration. Source-target mismatch enters mainly through the rows that map lifted pressure and input terms into \(w_{k+1}\): local gain, lag, hysteresis, drift, and bias alter the width dynamics without requiring a dense rewrite of the pressure state. A row-sparse change in these channels has low numerical rank in the lifted operator, motivating the structured correction.

For a candidate sequence \(\mathbf u=(u_{k|k},\ldots,u_{k+H-1|k})\), the controller rolls out the corrected predictor and scores
\begin{align}
J_k(\mathbf u)=&
\sum_{h=1}^{H}\!\left[q_w(\hat w_{k+h|k}-r_*)^2
+q_p(\hat p_{k+h|k}-p_*^r)^2\right]\nonumber\\
&+\sum_{h=0}^{H-1}\!\left(r_u u_{k+h|k}^2
+r_{\Delta u}\Delta u_{k+h|k}^2\right),
\end{align}
where \(\Delta u_{k+h|k}=u_{k+h|k}-u_{k+h-1|k}\). The candidates are subject to \(u_{k+h|k}\in\mathcal U_p(\hat p_{k+h|k})\) and \(\hat p_{k+h|k}\in[p_{\min},p_{\max}]\). The same cost and constraints score the shifted feasible plan, pressure-root pull, PI correction, zero move, sampled perturbation, and terminal fallback candidates in Algorithm~\ref{alg:akmpc}.

\begin{lemma}[Prediction tube under probe-conditioned memory error]
Assume the one-step width predictor satisfies
\[
|w_{k+1}-\hat w_{k+1|k}|\le \bar\ell |w_k-\hat w_{k|k}|+\bar\delta ,
\]
where
\begin{equation}
\bar\delta=\epsilon_f+\epsilon_y+\epsilon_b+L_a\|z_a-\hat z_a\|_2+L_\Delta\|\Delta\Theta-\widehat{\Delta\Theta}_{r_\Delta}\|_F .
\end{equation}
Then the \(h\)-step width error is bounded by \(E_h\) with
\begin{equation}
E_{h+1}=\bar\ell E_h+\bar\delta,\quad E_0=0,\quad
E_h\le \frac{1-\bar\ell^h}{1-\bar\ell}\bar\delta .
\end{equation}
\emph{Proof.} Apply the one-step inequality recursively. The probe-conditioned memory term enters only through \(\|z_a-\hat z_a\|_2\), which is reduced by the probe estimate and retrieval bound above.
\end{lemma}

The tube-filtered variant uses \(E_h\) as a scalar width-error tube. Pressure and input constraints remain hard. The width tube is used as a hard candidate screen only on horizons for which the tube margin is nonempty. Specifically, for \(E_h<\tau_w\), filtered-mode candidates must satisfy
\begin{equation}
\hat w_{k+h|k}\in[r_*-\tau_w+E_h,\ r_*+\tau_w-E_h].
\end{equation}
When \(E_h\ge\tau_w\), this interval is empty and the hard width screen is skipped. The implementation records a tube-overrun condition, ranks pressure/input-feasible candidates by the slack
\begin{equation}
s_h=\left[\,|\hat w_{k+h|k}-r_*|-(\tau_w-E_h)\,\right]_+,
\end{equation}
and applies the pressure-root or static inverse adaptive-PI fallback when the learned candidate does not reduce the slack.

\begin{theorem}[Finite-horizon tracking-cost degradation]
If \(|w_{k+h}-\hat w_{k+h|k}|\le E_h\) along a feasible candidate sequence, then the true tracking part of the finite-horizon cost satisfies
\begin{equation}
J_k^{true}(\mathbf u)\le J_k^{pred}(\mathbf u)+
\sum_{h=1}^{H}q_w\left(2|\hat w_{k+h|k}-r_*|E_h+E_h^2\right),
\end{equation}
with the same pressure-root and input penalties in both costs.
\emph{Proof.} Write \(w_{k+h}-r_*=(\hat w_{k+h|k}-r_*)+(w_{k+h}-\hat w_{k+h|k})\), expand the square, and use the bound \(E_h\).
\end{theorem}

In the experiments, prediction MAE estimates realized model error, rejection and fallback rates show how often the filter intervenes, input saturation and safety rows report pressure/input-limit stress, and paired tracking tables show whether smaller prediction and adapter errors become finite-horizon tracking gains. The no-memory and frozen-memory rows perturb the two explicit terms \(\|z_a-\hat z_a\|_2\) and \(\|\Delta\Theta-\widehat{\Delta\Theta}_{r_\Delta}\|_F\) in Lemma~2.

The commissioning run starts with a probe that estimates the pressure-width relation and local dynamic response (Fig.~\ref{fig:lifecycle}). The actuator-deadband-aware Koopman MPC then executes pressure-reference-aware MPC, updates only the local dynamic correction during the short run, and may preserve the best pressure-reference candidate for future warm starts. This retention step is an optional implementation refinement and is disabled in the no-retention arm.

\begin{figure}[!t]
\centering
\includegraphics[width=0.9\linewidth]{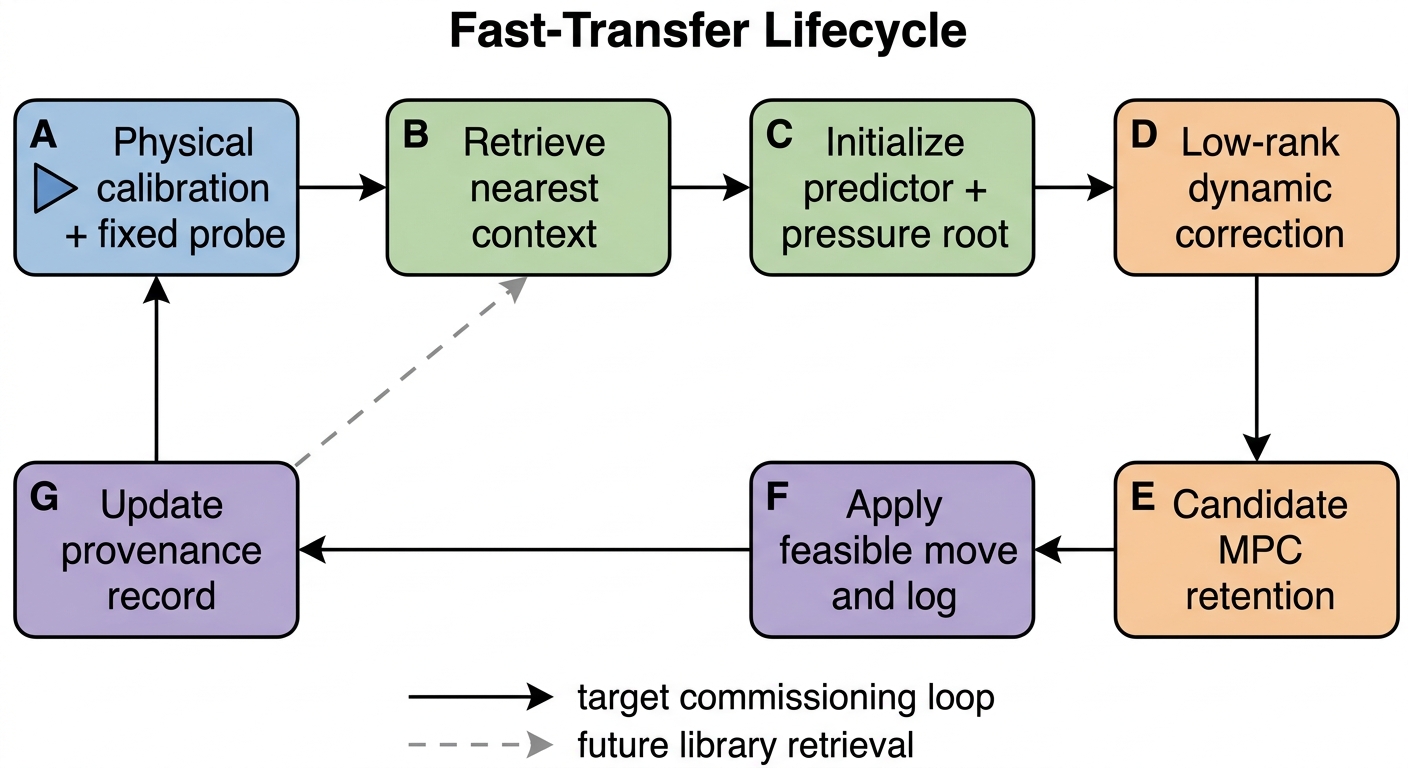}
\caption{Fast-transfer lifecycle from target probing to local correction, feasible control, and memory update.}
\label{fig:lifecycle}
\end{figure}

\section{Validation Protocol}

The main validation uses the low-width recipe \(r_b=1.00\) mm as the deadband case. Root-only pressure control is weak in this regime because part of the commanded pressure move is lost before it reaches the bead. The fixed 16-action probe exposes delivered-pressure efficiency and command hysteresis before feedback execution begins.

The benchmark contains 12 unseen target conditions and 5 seeds, giving 60 paired target-seed cases per controller. Each case uses the same start-pressure schedule, 15 episodes, 10 closed-loop benchmark steps per episode, MPC horizon 12, scaled pressure domain \([300,650]\), command bounds \([-15,15]\) scaled pressure units per step, and reporting band \(\tau_w=0.035\) mm. Controller constants are fixed across arms unless the arm explicitly removes a mechanism: adapter rank 4, retrieval weight \(\alpha=0.6\), 16 source contexts, and 16 noisy target static-calibration samples. These common settings keep the comparison from being confounded by different horizons, starts, calibration budgets, or information access. The target maps are low-width quadratic recipes mapped to the \(0.7\)--\(1.5\) mm operating range with a feasible \(1.00\) mm pressure root. The plant includes first-order bead-width lag, nonlinear command hysteresis, actuator deadband, reduced delivered-pressure gain, inspection noise, calibration-fit error, and slow material drift.

The controller comparison is performed on the calibrated digital-twin surrogate. The physical cell supplies the signal interface, calibration residuals, actuator limits, and anonymized response envelopes; the public benchmark supplies fixed seeds, target definitions, shared constraints, and reproducible controller logs.

The static descriptor is fitted from noisy plant-calibration samples before the surrogate closed-loop run, and the resulting calibration residual is stored with the target context. Static inverse pressure MPC is the primary calibration-aware baseline. Static inverse adaptive PI keeps the same inverse pressure target but corrects it online from measured width error. Probe-fitted ARX-MPC and online RLS-ARX-MPC provide dynamic-identification baselines. The probe-fitted ARX arm uses the same target probe but no retrieved memory, so it separates using probe data from reusing a historical PCM record. The remaining baselines do not receive the retrieved actuator descriptor; the comparison tests whether probe-conditioned memory improves the first commissioning run beyond calibration, probing, and online feedback alone. Mechanism arms remove probe-conditioned memory or freeze online target updating.

All controllers follow the same information-access rule. They may use noisy calibration samples, the fitted pressure root, the fixed probe when the arm uses probing, source context records when the arm uses memory, source-only distance normalizers, and online closed-loop measurements. Hidden target parameters, response-mode labels, and test-target statistics are excluded from controller initialization and are used only for surrogate-plant generation and post-run evaluation. The support package gives the full input schema and source-target split.

All 12 targets include deadband with delivered-pressure gain loss. They place that stress at different pressure roots and transient regimes. Table~\ref{tab:target-set} lists the root locations and dominant dynamic variations; Fig.~\ref{fig:probe-descriptor} shows how the fixed probe extracts the actuator descriptor used by memory lookup, and Fig.~\ref{fig:target-design} visualizes the same target spread.

\begin{table*}[!t]
\centering
\caption{Target-condition set for the \(1.00\) mm stress benchmark. Pressure root is reported in scaled pressure units.}
\label{tab:target-set}
\TableBodyFont
\begin{tabular}{@{}lrl@{\qquad}lrl@{}}
\toprule
Target & Root & Dominant variation & Target & Root & Dominant variation \\
\midrule
T01 & 462 & short lag, weak hysteresis & T07 & 454 & strong reversal hysteresis \\
T02 & 460 & longer lag, mild hysteresis & T08 & 465 & low delivered-pressure gain \\
T03 & 471 & high lag, high hysteresis & T09 & 476 & higher measurement noise, moderate lag \\
T04 & 482 & high static slope, command memory & T10 & 487 & high hysteresis, low delivered gain \\
T05 & 493 & upper-root operation, moderate drift & T11 & 459 & high lag, low drift \\
T06 & 443 & lower-root operation, moderate lag & T12 & 470 & combined lag, hysteresis, and drift \\
\bottomrule
\end{tabular}
\end{table*}

\begin{figure}[!t]
\centering
\includegraphics[width=0.92\linewidth]{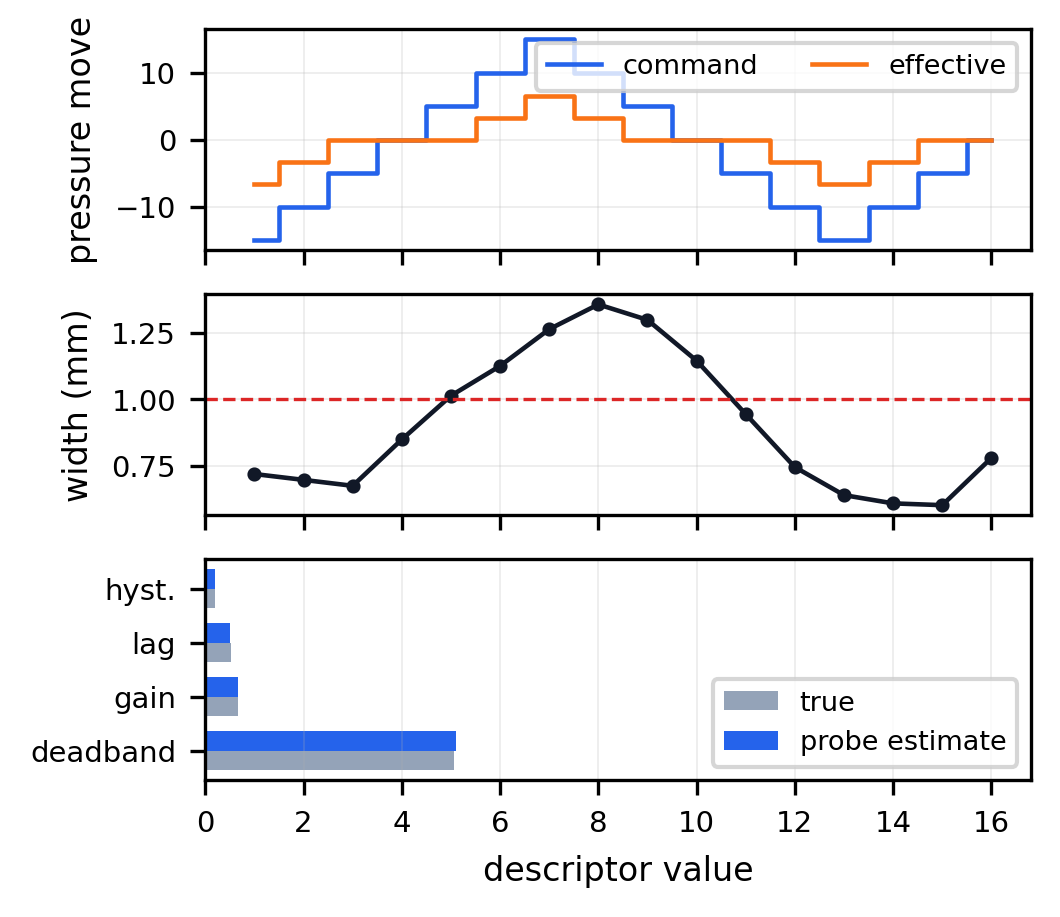}
\caption{Fixed 16-action probe and actuator descriptor construction for the \(1.00\) mm stress benchmark. The probe separates commanded and effective pressure moves and estimates deadband, delivered-pressure gain, lag, and hysteresis used by the retrieved context.}
\label{fig:probe-descriptor}
\end{figure}

\begin{figure}[!t]
\centering
\includegraphics[width=0.92\linewidth]{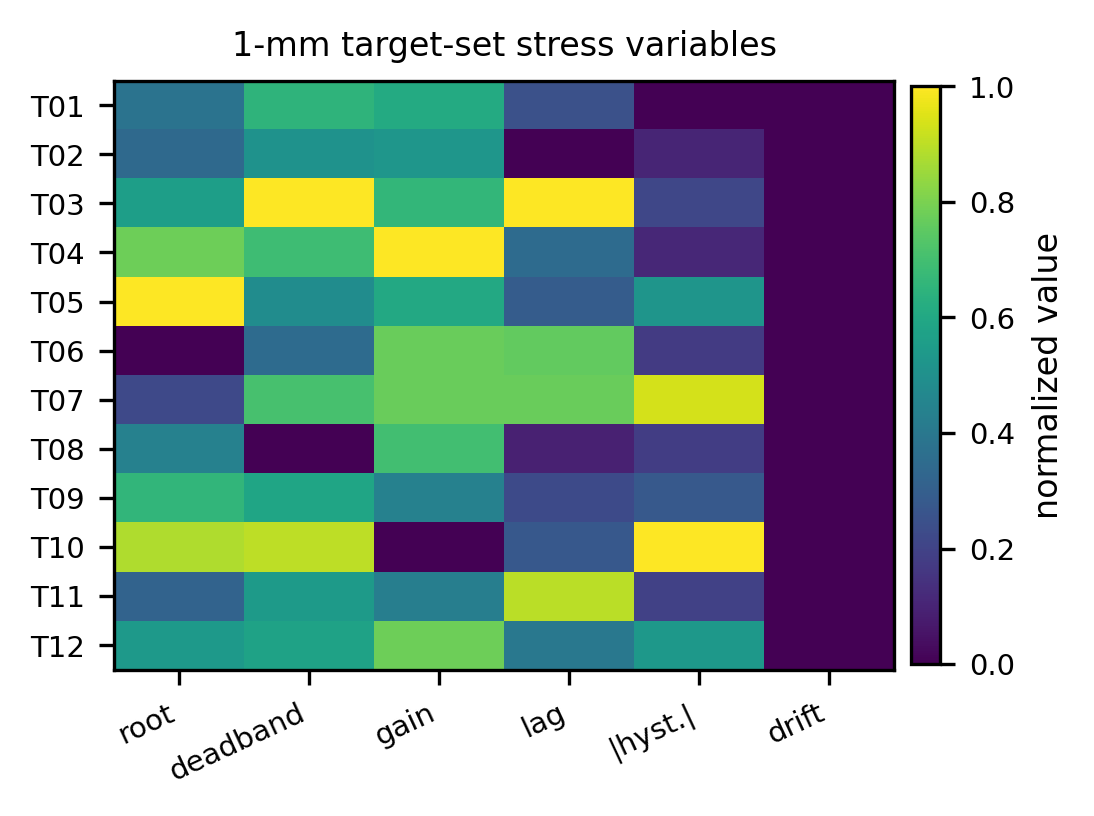}
\caption{Target-condition design for the 12 \(1.00\) mm stress-recipe targets. The heat map summarizes pressure-root location, actuator deadband, delivered-pressure gain, lag, hysteresis, and drift variations used in the paired benchmark.}
\label{fig:target-design}
\end{figure}

\section{Results and Deployment Behavior}

The \(1.00\) mm validation first compares tracking accuracy, then uses paired gains and ablations to separate probe-conditioned memory from ordinary online adaptation. Safety, smoothness, and runtime behavior are reported after the tracking results.

\begin{table*}[!t]
\centering
\caption{Main \(1.00\) mm actuator-deadband benchmark under the common information-access rule. Lower is better except runtime.}
\label{tab:fault-main}
\TableBodyFont
\setlength{\tabcolsep}{6pt}
\begin{tabular}{lrrrrr}
\toprule
Controller & Cases & Tracking MAE & Final-5 MAE & Pred. MAE & p95 ms \\
\midrule
AK-MPC & 60 & 0.0487 & 0.0442 & 0.0154 & 13.7 \\
AK-MPC without PCM & 60 & 0.0655 & 0.0705 & 0.0209 & 13.3 \\
Frozen-PCM AK-MPC & 60 & 0.1158 & 0.1693 & 0.0434 & 13.6 \\
Online RLS-ARX-MPC & 60 & 0.2492 & 0.2316 & 0.0395 & 16.7 \\
Static inverse pressure MPC & 60 & 0.2518 & 0.2662 & 0.2466 & 7.5 \\
Static inverse adaptive PI & 60 & 0.2665 & 0.2823 & 0.2655 & 7.5 \\
Probe-fitted ARX-MPC & 60 & 0.3956 & 0.3900 & 0.5090 & 16.5 \\
\bottomrule
\end{tabular}
\end{table*}

In Table~\ref{tab:fault-main} and Fig.~\ref{fig:fault-main-bars}, AK-MPC gives the lowest error in the deadband regime: \(0.0487\) mm tracking MAE and \(0.0442\) mm final-5 MAE. The comparison has two scales. Relative to classical calibration, feedback, and ARX baselines, the full AK-MPC workflow gives a large architecture-level gap because it combines constrained prediction, online correction, and feasibility filtering. Relative to the no-memory ablation, the isolated PCM gain is smaller but still visible: removing probe-conditioned memory raises tracking error to \(0.0655\) mm. The physical-cell trace in Fig.~\ref{fig:fault-representative} gives the corresponding deployed behavior for the T12/Pattern-111 stress case. Across increasing episode drift, measured actual width remains centered near the \(1.00\) mm recipe target and the learned predictor follows the same trend. The baseline gap follows the delivered-pressure loss seen in the fixed probe.

\begin{figure}[!t]
\centering
\includegraphics[width=0.92\linewidth]{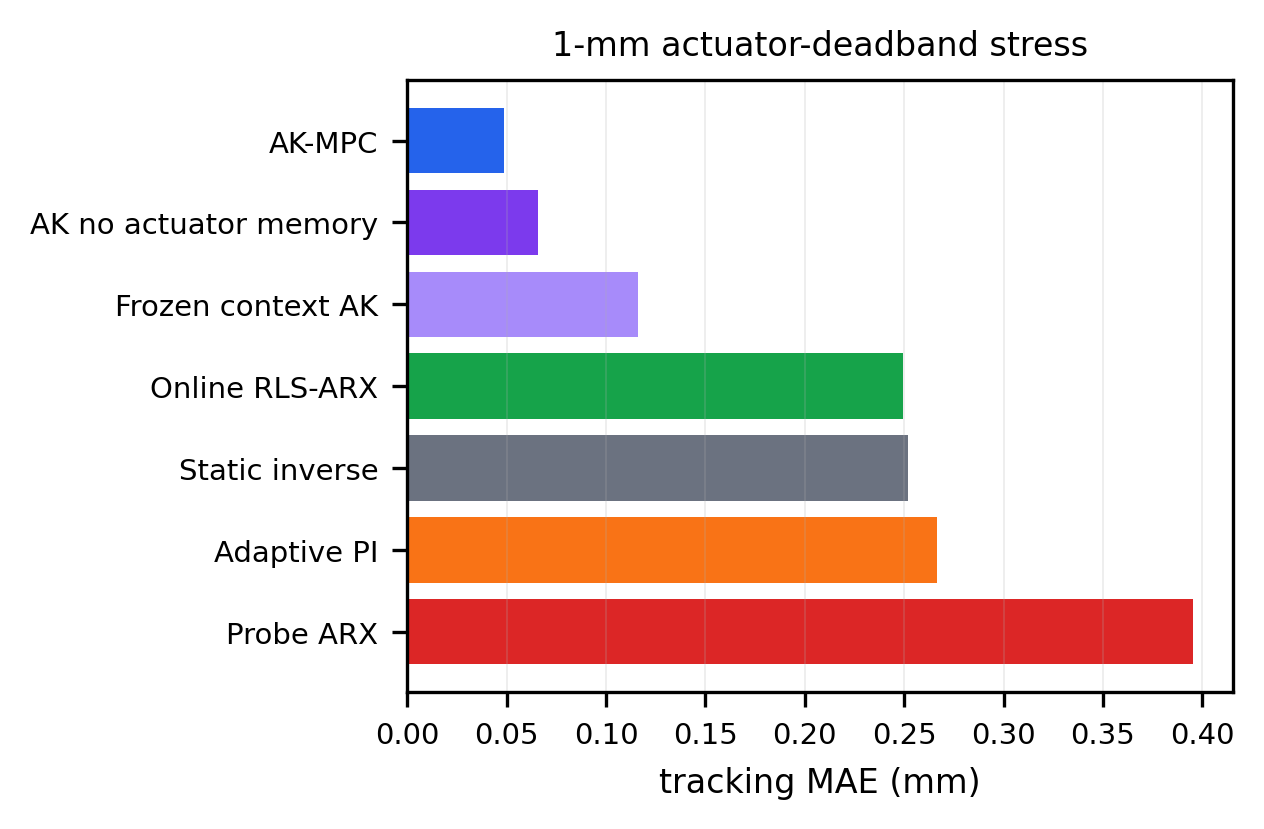}
\caption{Tracking MAE in the main \(1.00\) mm actuator-deadband benchmark. AK-MPC gives the lowest controller error, and removing probe-conditioned memory increases the error.}
\label{fig:fault-main-bars}
\end{figure}

\begin{figure*}[!t]
\centering
\includegraphics[width=0.92\textwidth]{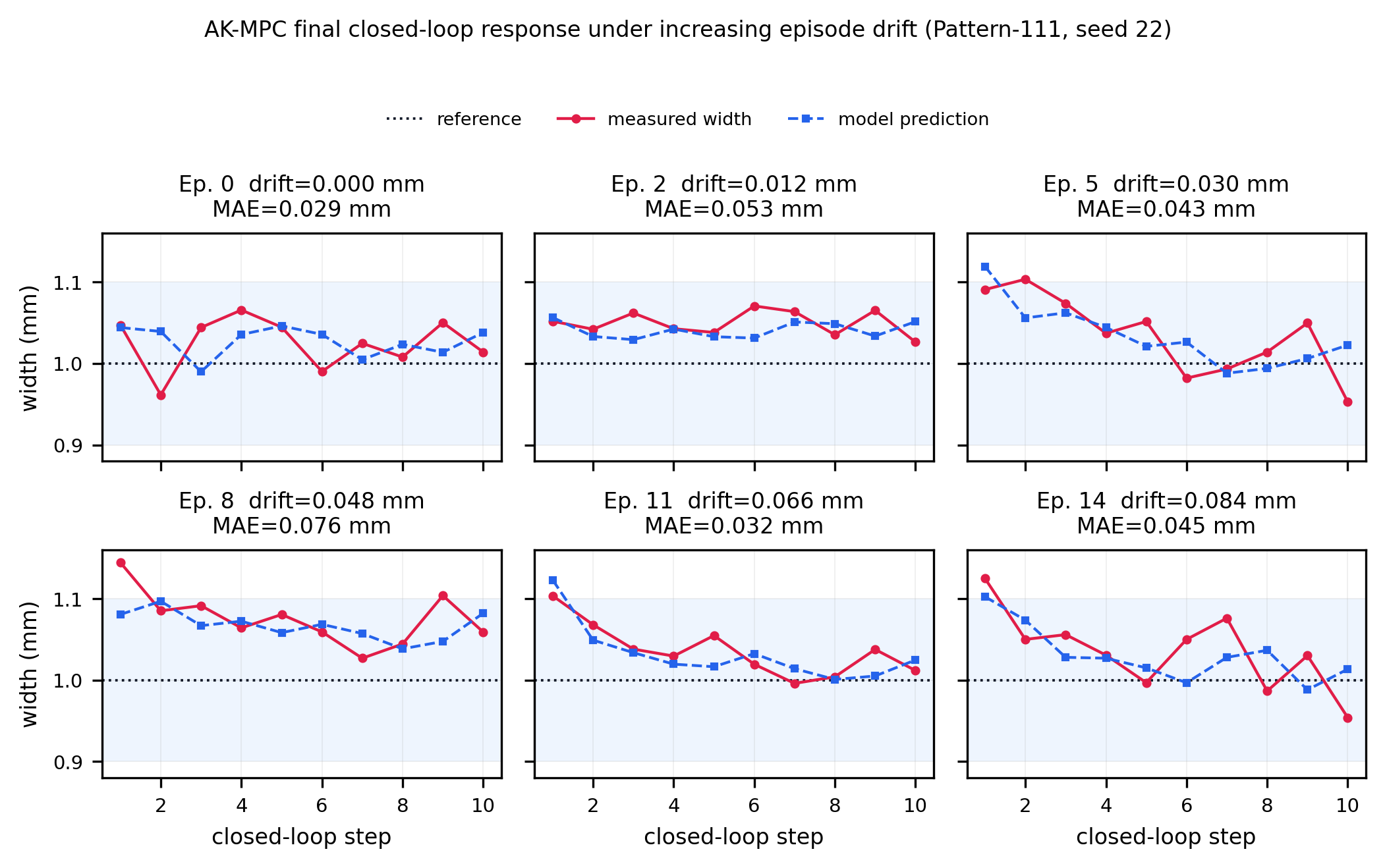}
\caption{Closed-loop control performance of the deployed AK-MPC on the physical sealing cell under increasing episode drift for the combined-dynamics target T12/Pattern-111. Red markers show the measured actual width from the vision camera, blue dashed markers show the one-step model prediction, the dotted line is the \(1.00\) mm reference, and the shaded band marks a \(\pm0.1\) mm control-effect range.}
\label{fig:fault-representative}
\end{figure*}

\begin{table*}[!t]
\centering
\caption{Paired tracking-MAE gains of AK-MPC over each baseline in the \(1.00\) mm actuator-deadband benchmark. Positive values mean AK-MPC has lower tracking error; CI denotes confidence interval.}
\label{tab:fault-paired}
\TableBodyFont
\begin{tabular*}{\textwidth}{@{\extracolsep{\fill}}lrrrr@{}}
\toprule
Baseline & Mean gain & Case-level 95\% CI & Target-clustered 95\% CI & Win rate \\
\midrule
Static inverse pressure MPC & 0.2031 & [0.1527, 0.2529] & [0.1417, 0.2627] & 100\% \\
Static inverse adaptive PI & 0.2178 & [0.1611, 0.2772] & [0.1488, 0.2826] & 100\% \\
Online RLS-ARX-MPC & 0.2005 & [0.1534, 0.2496] & [0.1444, 0.2600] & 100\% \\
Probe-fitted ARX-MPC & 0.3469 & [0.3078, 0.3858] & [0.3141, 0.3801] & 100\% \\
AK-MPC without PCM & 0.0168 & [0.0133, 0.0205] & [0.0135, 0.0203] & 86.7\% \\
Frozen-PCM AK-MPC & 0.0671 & [0.0608, 0.0739] & [0.0592, 0.0749] & 100\% \\
\bottomrule
\end{tabular*}
\end{table*}

Pairing by target and seed leaves the ordering unchanged. AK-MPC wins all 60 paired cases against the static inverse, adaptive PI, RLS-ARX, probe-ARX, and frozen-PCM baselines (Table~\ref{tab:fault-paired}). The target-clustered confidence intervals are positive throughout, and Fig.~\ref{fig:fault-target-gains} shows that the gains are not concentrated in a single target. The smaller but still positive gain over the no-memory ablation shows that probe-conditioned actuator terms in the retrieved context add information beyond the static map and the low-rank online correction.

\begin{figure}[!t]
\centering
\includegraphics[width=0.72\linewidth]{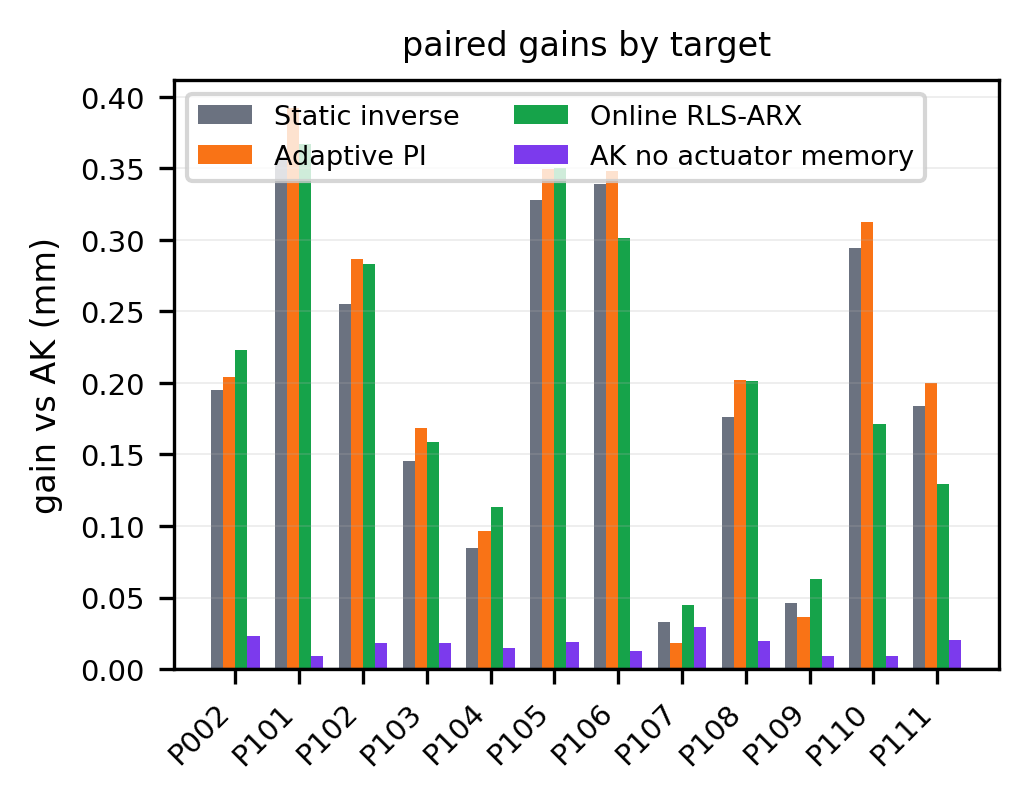}
\caption{Per-target paired gains; positive values indicate lower AK-MPC tracking MAE on matched cases.}
\label{fig:fault-target-gains}
\end{figure}

\begin{table*}[!t]
\centering
\caption{Mechanism comparisons in the \(1.00\) mm actuator-deadband benchmark. Each row removes or replaces one part of the workflow while keeping target-seed cases paired.}
\label{tab:mechanism-diagnostics}
\TableBodyFont
\setlength{\tabcolsep}{3pt}
\begin{tabular*}{\textwidth}{@{\extracolsep{\fill}}>{\raggedright\arraybackslash}p{0.24\textwidth}>{\raggedright\arraybackslash}p{0.30\textwidth}>{\centering\arraybackslash}p{0.08\textwidth}>{\centering\arraybackslash}p{0.20\textwidth}>{\centering\arraybackslash}p{0.08\textwidth}@{}}
\toprule
Mechanism tested & Comparison arm & AK gain & Target CI & Win rate \\
\midrule
Probe-conditioned memory & AK-MPC without PCM & 0.0168 & [0.0135, 0.0203] & 86.7\% \\
Online target update & Frozen-PCM AK-MPC & 0.0671 & [0.0592, 0.0749] & 100\% \\
Static pressure-root only & Static inverse pressure MPC & 0.2031 & [0.1417, 0.2627] & 100\% \\
Static inverse feedback & Static inverse adaptive PI & 0.2178 & [0.1488, 0.2826] & 100\% \\
Generic online ARX dynamics & Online RLS-ARX-MPC & 0.2005 & [0.1444, 0.2600] & 100\% \\
Probe-only ARX dynamics & Probe-fitted ARX-MPC & 0.3469 & [0.3141, 0.3801] & 100\% \\
\bottomrule
\end{tabular*}
\end{table*}

The ablations localize the gain. Removing probe-conditioned memory increases tracking error by \(0.0168\) mm, and freezing the online update costs \(0.0671\) mm (Table~\ref{tab:mechanism-diagnostics}). The mechanism plot in Fig.~\ref{fig:fault-mechanism} shows the same ordering. Static pressure-root control, adaptive static feedback, and generic ARX/RLS dynamics are weaker here because the short-probe delivered-pressure loss is absent from their initialization.

\begin{figure}[!t]
\centering
\includegraphics[width=0.80\linewidth]{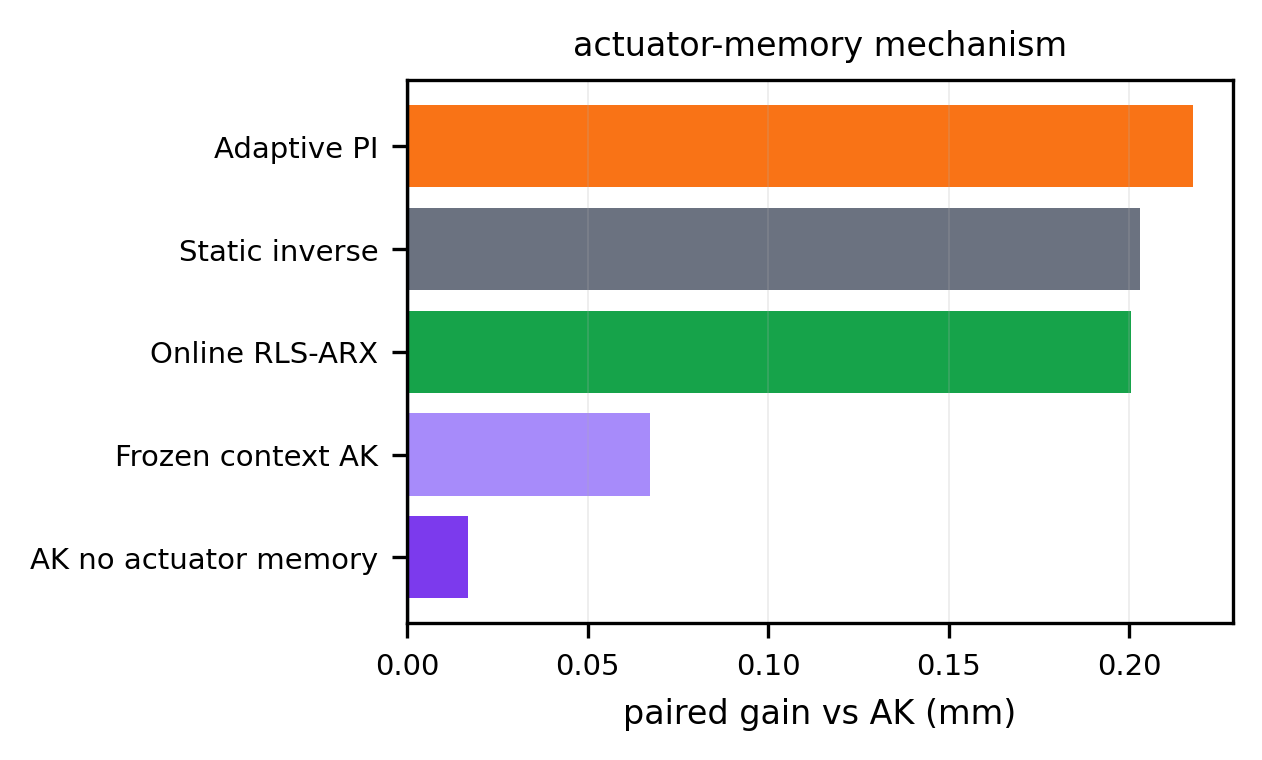}
\caption{Mechanism ablation results for actuator-deadband validation. Positive paired gains correspond to lower AK-MPC tracking MAE after removing or replacing the listed mechanism.}
\label{fig:fault-mechanism}
\end{figure}

The main tables use only information available at deployment: physical calibration summaries, source contexts, the fixed target probe, online closed-loop observations, and recipe-level bounds. Hidden target coefficients, response-mode labels, test-target summary statistics, and future closed-loop observations are excluded from controller initialization and retrieval.

\begin{table*}[!t]
\centering
\caption{Practical scale, safety, smoothness, and runtime results in the \(1.00\) mm actuator-deadband benchmark. The out-of-band rate uses \(\tau_w=0.035\) mm as a reporting band; TV denotes total variation.}
\label{tab:fault-safety}
\TableBodyFont
\setlength{\tabcolsep}{3pt}
\begin{tabular*}{\textwidth}{@{\extracolsep{\fill}}>{\raggedright\arraybackslash}p{0.27\textwidth}rrrrrrr@{}}
\toprule
Controller & Press. viol. & Input viol. & Out-of-band & Input sat. & Press. TV & Act. TV & p95 ms \\
\midrule
AK-MPC & 0 & 0 & 52.1\% & 0.0\% & 540.6 & 32.5 & 13.7 \\
AK-MPC without PCM & 0 & 0 & 73.1\% & 0.8\% & 623.3 & 59.5 & 13.3 \\
Frozen-PCM AK-MPC & 0 & 0 & 88.3\% & 0.0\% & 506.2 & 27.9 & 13.6 \\
Online RLS-ARX-MPC & 0 & 0 & 92.3\% & 3.6\% & 694.2 & 121.2 & 16.7 \\
Static inverse pressure MPC & 0 & 0 & 86.2\% & 0.0\% & 522.4 & 45.9 & 7.5 \\
Static inverse adaptive PI & 0 & 0 & 84.3\% & 0.3\% & 597.3 & 97.9 & 7.5 \\
Probe-fitted ARX-MPC & 0 & 0 & 97.8\% & 27.1\% & 1331.2 & 180.4 & 16.5 \\
\bottomrule
\end{tabular*}
\end{table*}

No controller violates pressure or input bounds in the logged execution. Table~\ref{tab:fault-safety} shows that AK-MPC has the lowest out-of-band rate among controllers that use only deployment-time information and no input saturation in this run. Fig.~\ref{fig:fault-feasibility} provides the corresponding constraint/runtime heat map. The probe-conditioned-memory controller also uses less action total variation than static inverse adaptive PI and online RLS-ARX-MPC, consistent with its explicit model of deadband and delivered-pressure gain. Runtime was measured on a Windows 11 workstation with an Intel Core Ultra 5 125H central processing unit (CPU), Python 3.8.10, NumPy 1.24.4, and CPU-only PyTorch 2.4.1; the reported p95 values cover candidate rollout, prediction, scoring, and filtering, and exclude plotting and comma-separated-value (CSV) export.

The error scale differs from a root-only benchmark. The physical calibration RMSE remains \(0.010\) mm, while the \(1.00\) mm run adds a target-fit residual and an actuator-delivery error that static inverse feedback cannot infer from the pressure root alone. AK-MPC reduces prediction MAE to \(0.0154\) mm and keeps the lowest tracking MAE, whereas the static and ARX baselines either miss the delivered-pressure loss or identify it too late within the short commissioning budget.

\begin{figure}[!t]
\centering
\includegraphics[width=0.78\linewidth]{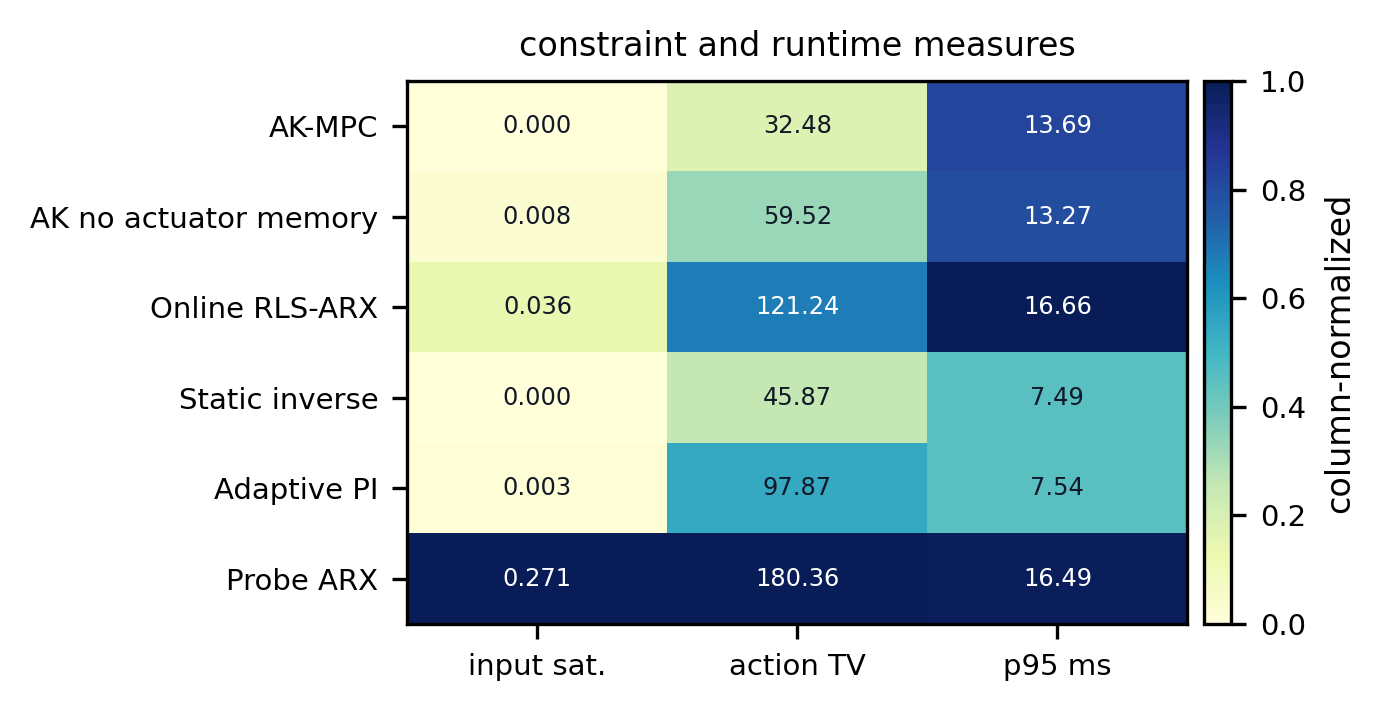}
\caption{Constraint and runtime heat map for the \(1.00\) mm actuator-deadband benchmark, with values colored by column-normalized scale.}
\label{fig:fault-feasibility}
\end{figure}

\subsection{Auxiliary Runs}

Earlier sensitivity, retrieval-stability, and cross-recipe runs are retained in the support material as auxiliary validation runs. These auxiliary files vary the probe length over 8, 16, and 24 moves, the adapter rank over 2, 4, and 8, the retrieval weight \(\alpha\) over 0.3, 0.6, and 0.9, and the source-library size over 4, 8, 16, and 32 records under the same 12-target, 5-seed, horizon-12 organization. The main benchmark reported above uses the same target, seed, start-pressure, probe, horizon, pressure-bound, and input-bound schedule for every arm. The exported files contain 420 controller-cases, 63000 closed-loop surrogate step rows, zero non-finite case metrics, zero pressure violations, and zero input violations.

Those auxiliary runs bound the scope of the \(1.00\) mm result by repeating the run organization under related sensitivity, retrieval, and recipe settings.

Detailed figure/table provenance is kept in the support folder, leaving the main article focused on the control mechanism and the controller comparison.

\section{Discussion}

The results are strongest in the regime for which the benchmark was built: a narrow bead recipe whose calibrated pressure root remains valid but whose delivered pressure is history dependent. In that setting, the probe gives the controller an initial estimate of actuator loss, while the online correction removes the remaining target-fit residual during the episode sequence. The no-memory and frozen-memory rows separate these two roles: the stored actuator terms improve the first run, and target updating keeps the predictor aligned as drift accumulates. This distinction is important for interpreting the numbers: the large gap to classical baselines is an AK-MPC workflow effect, while the no-memory row quantifies the additional contribution of retrieved PCM.

The same mechanism also shows where these gains disappear. If a recipe behaves like a root-only static map, static inverse feedback may be sufficient. If the probe does not excite the relevant deadband direction, or if the source library lacks a nearby delivery mode, retrieval can only provide a weak prior and the fallback filter becomes more important. The current probe/no-memory midpoint is the probe-fitted ARX controller; a hand-designed deadband-compensated PI controller would be a useful additional midpoint for a larger physical campaign. The method is aimed at low-width recipes where delivered-pressure loss is visible before feedback execution and has appeared in previous operating records.

The broad paired comparison remains a calibrated digital-twin benchmark, while Fig.~\ref{fig:fault-representative} provides a physical-cell closed-loop trace for the T12/Pattern-111 stress case. This trace should be read as a deployment sanity check rather than as a 60-case physical benchmark. It confirms that the retrieved probe-conditioned memory and online correction handle true process disturbances, including pressure loss and material drift, on the measured signal path.

\FloatBarrier

\section{Conclusion}

This study developed probe-conditioned memory for actuator-deadband-aware Koopman MPC in low-width sealing recipes. The \(1.00\) mm benchmark isolates the failure of root-only pressure control: the static pressure setpoint is calibrated, but delivered pressure depends on recent actuator motion. AK-MPC gives the lowest tracking error across the paired digital-twin comparisons, and the ablations show that both the stored actuator terms and the online target update contribute to the observed error reduction. The physical-cell trace supports the same mechanism on one industrial stress case, but the statistical benchmark remains the calibrated surrogate study. The practical message is narrower than a general transfer-learning claim: historical runs help most when they carry pressure-delivery information that a short target probe can verify and refine.

\appendices
\section{Implementation Details}

The main actuator-deadband results are generated by the \(1.00\) mm benchmark runner. The run contains 12 targets, 5 seeds, 15 episodes, 10 steps per episode, horizon 12, and 7 controller arms, giving 420 controller-cases and 63000 closed-loop surrogate step rows. The package contains the corresponding case summaries, paired comparisons, safety/runtime summaries, generated figure sources, and reproducibility files; earlier sensitivity and cross-recipe outputs remain in the support material as auxiliary validation runs. Source-library growth is handled by nearest-neighbor retrieval over the stored descriptor tuple in this implementation; large-scale record pruning or lifecycle-based retirement is left to plant deployment policy rather than claimed as a solved contribution here.

\FloatBarrier

\section*{Data and Code Availability}

Raw physical-cell and production records are not publicly released because they contain confidential product-recipe, material, and process-operation information. The released package contains anonymized physical calibration summaries, probe-response summaries, actuator and measurement bounds, target-condition definitions, digital-twin generation scripts, controller configurations, surrogate step histories, paired-comparison tables, information-access files, and figure/table provenance files. These artifacts reproduce the reported digital-twin controller comparisons without exposing proprietary industrial records. The main benchmark outputs are stored in the package support folder and matching project data directory; older sensitivity and cross-recipe outputs are retained as auxiliary results.

\end{document}